\documentclass[letterpaper, 10 pt, conference]{ieeeconf}
\IEEEoverridecommandlockouts                              % This command is only
\usepackage{amsmath} % assumes amsmath package installed
\usepackage{amssymb}  % assumes amsmath package installed
\usepackage{cite}
\usepackage{booktabs} 
\usepackage{xcolor}
\usepackage{amsfonts}
\usepackage{subcaption}
\usepackage{graphicx}      % include this line if your document contains figures
\usepackage{float}
\usepackage[fancythm,fancybb]{jphmacros2e} 
\usepackage{algorithm}
\usepackage{algpseudocode}
\usepackage{enumerate}
\usepackage{balance}
\renewcommand{\QED}{\QEDopen}

\title{\LARGE \bf
Failure-Aware Iterative Learning of State-Control Invariant Sets
}

\author{Ahmad Amine, Nick-Marios T. Kokolakis, Ugo Rosolia, Truong X. Nghiem, and Rahul Mangharam% <-this % stops a space
\thanks{This work was partially supported by US DoT Safety21 National University Transportation Center and NSF awards CISE-2431569 and 2514584.}% <-this % stops a space
\thanks{A. Amine, N.-M. T. Kokolakis and R. Mangharam are with the Department of Electrical and Systems Engineering, University of Pennsylvania, Philadelphia, PA 19104, USA
        {\tt\small \{aminea, nmkoko\}@seas.upenn.edu}}%
\thanks{U. Rosolia is with Lyric, Italia 152 Catania Italy. E-mail: ugo.rosolia@gmail.com}%
\thanks{T. X. Nghiem is with the Department of Electrical and Computer Engineering, University of Central Florida, Orlando, FL 32816, USA. E-mail: truong.nghiem@ucf.edu.}%
\thanks{This work has been submitted to the IEEE for possible publication. Copyright may be transferred without notice, after which this version may no longer be accessible.}
}

\begin{document}

\maketitle
\thispagestyle{empty}
\pagestyle{empty}

%%%%%%%%%%%%%%%%%%%%%%%%%%%%%%%%%%%%%%%%%%%%%%%%%%%%%%%%%%%%%%%%%%%%%%%%%%%%%%%%
\begin{abstract}

In this paper, we address the problem of computing maximal state-control invariant sets for deterministic linear systems using failing trajectories. We introduce the concept of state-control invariance, which extends control invariance from the state space to the joint state-control space. The maximal state-control invariant (MSCI) set simultaneously encodes the maximal control invariant set (MCI) and, for each state in the MCI, the set of control inputs that preserve invariance. We prove that the state projection of the MSCI is the MCI and the state-dependent sections of the MSCI are the admissible invariance-preserving inputs. Building on this framework, we develop a Failure-Aware Iterative Learning (FAIL) algorithm for deterministic linear time-invariant systems with polytopic constraints. The algorithm iteratively updates a constraint set in the state-control space by learning predecessor halfspaces from one-step failing state-input pairs, without knowing the dynamics. For each failure, FAIL learns the violated halfspaces of the predecessor of the constraint set by a regression on failing trajectories. We prove that the learned constraint set converges monotonically to the MSCI. Numerical experiments on a double integrator system validate the proposed approach.

\end{abstract}

\section{Introduction}
\label{sec:intro}

Safety-critical control systems require the state to remain within a prescribed set of admissible states at all times. The maximal control invariant (MCI) set $\mathcal{X}_\infty$ characterizes the largest set of states from which admissible inputs can satisfy constraints indefinitely~\cite{blanchini1999set}. MCI sets are used as terminal constraints in model predictive control, ensuring recursive feasibility~\cite{borrelli2017predictive} and safety~\cite{zang2025sitlmpc}, and as certificates in safety filters for learning-based controllers~\cite{ames2019control}.
 
Classical algorithms compute the MCI set via the one-step predecessor operator, which iteratively removes states from which no admissible input can keep the system within the constraint set~\cite{borrelli2017predictive}. These algorithms terminate in finitely many steps for linear systems with polytopic constraints~\cite{blanchini1999set}. Extensions to nonlinear systems employ sum-of-squares programming~\cite{korda2014convex} and neural certificates~\cite{kokolakis2025safe}. Nevertheless, all of these methods require an accurate dynamics model, which may be unavailable or expensive to obtain.
 
Data-driven methods address this model dependence. Behavioral approaches enable controller design with invariance guarantees for linear systems~\cite{depersis2020formulas}. Direct data-driven computation of invariant sets has been explored for both linear~\cite{chen2018data} and nonlinear systems~\cite{KASHANI2025112010}. These methods compute only the set of safe states $\mathcal{X}_\infty$, and not which control inputs at each state preserve invariance.
 
Recent work addresses this limitation by considering sets in the joint state-control space. State-action control barrier functions~\cite{he2023sacbf} define safety in the joint state-control space by evaluating both the state and the applied input, yielding a convex safety filter with reduced online cost. In reinforcement learning and imitation learning, failures instead serve as negative samples that shape rewards, constrain value functions, or bias exploration~\cite{shiarlis2016inverse, bharadhwaj2021conservative, esteban2018learning}. Neither line of work produces explicit constraint sets that provably prevent the recurrence of observed failures.
 
To the best of our knowledge, an approach that iteratively learns explicit invariant constraint sets in the joint state-control space from observed failures is absent from the literature.

The contributions of this paper are twofold. First, we introduce joint state-control invariance for discrete-time controlled nonlinear time-invariant dynamical systems and prove that the maximal state-control invariant (MSCI) set recovers both the MCI and the admissible invariance-preserving inputs. We derive a predecessor operator in the joint state-control space and establish convergence of a recursive algorithm for deriving the MSCI. Second, we consider deterministic linear time-invariant (LTI) systems with polytopic constraints and develop a failure-aware iterative learning (FAIL) algorithm that learns the halfspaces that define the state-control invariant set from observed one-step failing state-input pairs. The algorithm converges monotonically to the MSCI without system identification, requiring only failing trajectories with sufficient excitation.

The rest of this paper is organized as follows. Section~\ref{sec:problem} formulates the learning-from-failure problem. Section~\ref{sec:state-control-invariant} introduces joint state-control invariance. Section~\ref{sec:learning} presents the failure-aware learning algorithm. Section~\ref{sec:results} provides numerical validation, and Section~\ref{sec:conclusion} concludes the paper.

\paragraph*{Notation} We denote by $\mathbb{Z}^+$ and $\overline{\mathbb{Z}}_{+}$ the sets of positive and non-negative integers respectively. For a set $S \subseteq \mathbb{R}^n$, $2^{S}$ denotes its power set. We use $\triangleq$ for definitional equalities. We use $(C)_j\in \mathbb{R}^{1 \times m}$ to denote the $j$th row of a matrix $C \in \mathbb{R}^{n\times m}$.

\section{Problem Formulation}
\label{sec:problem}
In this section, we state the problem of using failing trajectories to learn the maximal control invariant set and the control inputs that ensure invariance at each state in this set.

Consider the discrete-time controlled nonlinear time-invariant dynamical system given by
\begin{align}
\label{eq:dyn}
x(k+1)=f(x(k), u(k)), \quad x(0) = x_0, \quad\forall k\in\overline{\mathbb{Z}}_{+},
\end{align}
where $x(k)\in\mathbb{R}^{n_x}$ is the state and $u(k)\in\mathbb{R}^{n_u}$ is the control input at time step $k$.
The mapping $f : \mathbb{R}^{n_x} \times \mathbb{R}^{n_u} \to \mathbb{R}^{n_x}$ is Lipschitz continuous with $f(0,0) = 0$. The state is constrained to a compact set $\mathcal{X}\subset\mathbb{R}^{n_x}$ containing the origin in its interior, and the control input is constrained to a compact set $\mathcal{U}(x)\subseteq \mathbb{R}^{n_u}$, which depends on the state $x \in \mathbb{R}^{n_x}$, and satisfies $0 \in \mathcal{U}(0)$.

To state the learning-from-failure problem, we introduce the following definitions that introduce the concepts of \emph{control invariance}, \emph{maximal control invariance}, and \emph{safety}.

\begin{definition}[Control invariant set]\label{def:ci}
Consider system \eqref{eq:dyn} with state and input constraints $\mathcal{X}$ and $\mathcal{U}(x)$ for all $x \in \mathbb{R}^{n_x}$. A set $\mathcal{C}\subseteq\mathcal{X}$ is called \emph{control invariant} if, for every $x\in\mathcal{C}$, there exists $u\in\mathcal{U}(x)$ such that $f(x, u)\in\mathcal{C}$.
\end{definition}

\begin{definition}[Maximal control invariant set]\label{def:mci}
A set $\mathcal{X}_\infty\subseteq\mathcal{X}$ is called the \emph{MCI} if
\begin{enumerate}[(i)]
\item $\mathcal{X}_\infty$ is control invariant, and
\item if $\mathcal{C}\subseteq\mathcal{X}$ is any control invariant set, then $\mathcal{C}\subseteq\mathcal{X}_\infty$.
\end{enumerate}
\end{definition}
\begin{definition}[Safety]\label{def:safety}
A set $\mathcal{C}\subseteq\mathcal{X}$ is \emph{safe} if it is a control invariant set with respect to~\eqref{eq:dyn} and $\mathcal{U}(x)$ for all $x \in \mathcal{C}$.
\end{definition}
It follows from Definitions~\ref{def:mci} and~\ref{def:safety} that, if the initial state $x(0) \in \mathcal{X}_\infty$, then there exists a sequence of admissible inputs that keeps the state in $\mathcal{X}$ indefinitely. Hence, $\mathcal{X}_\infty$ is safe.

The MCI set can be computed via the \emph{one-step predecessor operator}. For any set $\Omega \subseteq \mathbb{R}^{n_x}$, define
\begin{equation}\label{eq:pre-state}
    \operatorname{Pre}(\Omega) \triangleq \big\{ x \in \mathbb{R}^{n_x} : \exists\, u \in \mathcal{U}(x) \;\text{s.t.}\; f(x, u) \in \Omega \big\},
\end{equation}
as the set of all states from which there exists an admissible input that steers the system into $\Omega$ in one step. The MCI set can be obtained as the fixed point of the recursive iteration $\Omega_{k+1} = \operatorname{Pre}(\Omega_k) \cap \Omega_k$ with $\Omega_0 = \mathcal{X}$ and $k \in \overline{\mathbb{Z}}_{+}$, where, under the compactness and continuity assumptions, it is ensured that $\lim_{k \to \infty} \bigcap_{i=0}^k \Omega_i$ is non-empty and $\mathcal{X}_\infty = \lim_{k \to \infty} \bigcap_{i=0}^k \Omega_i$ ~\cite{bertsekas1972infinite}. For deterministic linear time-invariant (LTI) systems with polytopic constraints, the recursion terminates in a finite number of steps~\cite{borrelli2017predictive}.

It is important to note that the MCI set only characterizes \emph{where} the system \emph{can} remain safe but not \emph{which} inputs at each state ensure invariance. The existential quantifier in~\eqref{eq:pre-state} discards all control information during the recursion. 

Given the dynamics $f$ and the MCI $\mathcal{X}_\infty$, one can design a state-feedback controller $\pi\colon\mathcal{X}\to \mathbb{R}^{n_u}$ such that $\pi(x) \in \mathcal{U}(x)$ and $f(x, \pi(x))\in\mathcal{X}_\infty$ for every $x\in\mathcal{X}_\infty$ \cite{kerrigan2001robust}, thereby guaranteeing safety. However, when the system dynamics are unknown, a model-based controller must rely on an approximate model $\hat{f}$. In this case, for a given $x \in \mathcal{X}$, such a controller may select an input $u \in \mathcal{U}(x)$ that drives the next state $f(x, u)$ outside $\mathcal{X}_\infty$, resulting in \emph{failure}. 

Next, we define the notions of a \textit{one-step failing state-input pair} and a \textit{failing trajectory}.

\begin{definition}[One-step failing state-input pair]\label{def:one-step-failure}
Consider system~\eqref{eq:dyn} with state and input constraints $\mathcal{X}$ and $\mathcal{U}(x)$ for all $x \in \mathbb{R}^{n_x}$. A state-input pair $(x, u)$ with $x \in \mathcal{X}$ and $u\in \mathcal{U}(x)$ is called a \emph{one-step failing state-input pair} if $f(x, u) \notin \mathcal{X}$.
\end{definition}

\begin{definition}[Failing trajectory]\label{def:failing-trajectory}
Consider system~\eqref{eq:dyn} with state and input constraints $\mathcal{X}$ and $\mathcal{U}(x)$ for all $x \in \mathcal{X}$, and a state-feedback controller $\pi: \mathcal{X} \to \mathbb{R}^{n_u}$ satisfying $\pi(x) \in \mathcal{U}(x)$ for all $x \in \mathcal{X}$. Given $x(0) \in \mathcal{X}$ and $L \in \mathbb{Z}^+$, a trajectory $\mathcal{T} \triangleq \bigl\{(x(k),\, u(k))\bigr\}_{k=0}^{L-1}$, where $u(k) = \pi(x(k))$ and $x(k+1) = f(x(k), u(k))$, is a \textit{failing trajectory} if $(x(L-1), u(L-1))$ is the only one-step failing state-input
pair in $\mathcal{T}$.
\end{definition}

A one-step failing state-input pair reveals that although the input is in $\mathcal{U}(x)$, it yields a next state that is not in $\mathcal{X}$. To learn from such events and prevent their recurrence, we need to characterize not only $\mathcal{X}_\infty$ but also the set of inputs that keep the system within $\mathcal{X}_\infty$. Crucially, this should be accomplished without knowing $f$. 

Now, we define our \textit{learning-from-failure} problem.
\begin{problem}[Learning-from-failure]\label{prob:main}
Consider the dynamical system~\eqref{eq:dyn} with state constraints $\mathcal{X}$ and input constraints $\mathcal{U}(x)$ for all $x \in \mathbb{R}^{n_x}$. Let $\mathcal{X}_\infty\subseteq\mathcal{X}$ be the MCI set with respect to~\eqref{eq:dyn} and $\mathcal{U}(x)$ for all $x \in \mathbb{R}^{n_x}$. Assume that $f$ and $\mathcal{X}_\infty$ are unknown. Given trajectories that include one-step failing state-input pairs, determine the set $\mathcal{X}_\infty$ and, for each $x \in \mathcal{X}_\infty$, determine the set of control inputs $u$ in $\mathcal{U}(x)$, such that $f(x,u)$ is in $\mathcal{X}_\infty$.
\end{problem}

Throughout, a \emph{failure} is the occurrence of a one-step failing state-input pair rather than a fault of a system component, and \emph{learning} refers to identifying $\mathcal{X}_\infty$ from state-control samples rather than identifying a model of 
$f$.

\section{Joint State-Control Invariance}
\label{sec:state-control-invariant}

This section introduces the concept of \textit{joint state-control invariant sets} and develops an algorithm to compute this set assuming that the dynamics $f$ is known. This assumption will be relaxed in Section~\ref{sec:learning} where we develop an algorithm for learning from failures for deterministic LTI systems.

\subsection{From invariant sets to invariance-preserving inputs}
\label{subsec:invariance-preserving-inputs}

To discern which control inputs at each state ensure invariance, we seek a characterization of the control inputs that render a given set of states safe.

We now define the \textit{invariance-preserving input set}.

\begin{definition}[Invariance-preserving input set]\label{def:inv-input}
Consider system \eqref{eq:dyn} and a set $C \subseteq \mathbb{R}^{n_x}$.
The \emph{invariance-preserving input set} of $C$ under dynamics $f$ is the set-valued map
$\mathcal{U}^f_{C} : \mathbb{R}^{n_x} \to 2^{\mathbb{R}^{n_u}}$
defined by
\begin{equation} \label{eq:policy-req}
    \mathcal{U}^f_{C}(x) \triangleq \big\{u\in \mathbb{R}^{n_u}
    : f(x,u)\in C\big\}, \quad x\in\mathbb{R}^{n_x}.
\end{equation}
\end{definition}
It follows from Definition~\ref{def:inv-input} that for a given state $x \in \mathbb{R}^{n_x}$, $\mathcal{U}^f_{C}(x)$ is the set of control inputs under
which the next state is in $C$. For convenience, we write $\mathcal{U}^f_\infty(x) \triangleq
\mathcal{U}^f_{\mathcal{X}_\infty}(x)$.

Given an initial state $x(0) \in \mathcal{X}_\infty$, a controller that selects $u(k)$ from
$\mathcal{U}^f_\infty(x(k)) \cap\, \mathcal{U}(x(k))$ for all $k \in \overline{\mathbb{Z}}_{+}$ ensures that $x(k)$ remains in $\mathcal{X}_\infty$. We prove this result in the following lemma.
\begin{lemma}
\label{lemma:admissible-input-to-invariance}
Consider %the time-invariant discrete-time dynamical 
system \eqref{eq:dyn}
with state and input constraints $\mathcal{X} \subset \mathbb{R}^{n_x}$ and
$\mathcal{U}(x) \subset \mathbb{R}^{n_u}$ for all $x \in \mathbb{R}^{n_x}$. Let
$\mathcal{X}_\infty \subseteq \mathcal{X}$ be the MCI
set and $\pi: \mathbb{R}^{n_x} \to \mathbb{R}^{n_u}$ be a feedback controller such that
$\pi(x) \in \mathcal{U}^f_\infty(x) \cap \mathcal{U}(x)$ for all
$x \in \mathcal{X}_\infty$ where $\mathcal{U}^f_\infty(x) \cap \mathcal{U}(x)$ is assumed to be non-empty. If $x(0) \in \mathcal{X}_\infty$, then under $\pi$, the solution sequence $x(k),\ k \in \overline{\mathbb{Z}}_{+},$ to \eqref{eq:dyn}
satisfies $x(k) \in \mathcal{X}_\infty$ and $u(k) \in \mathcal{U}(x(k))$ for all $k \in \overline{\mathbb{Z}}_{+}$.
\end{lemma}
\begin{proof}
Note that $x(0) \in \mathcal{X}_\infty$ holds by
assumption. Suppose $x(k) \in \mathcal{X}_\infty$ for some $k \in \overline{\mathbb{Z}}_{+}$.
Since $\pi(x(k)) \in \mathcal{U}^f_\infty(x(k))$, it follows from
\eqref{eq:policy-req} that
$x(k+1) = f(x(k), \pi(x(k))) \in \mathcal{X}_\infty$. Moreover,
$\pi(x(k)) \in \mathcal{U}(x(k))$ ensures that the applied input is
admissible. Now, the result follows by induction.\frQED
\end{proof}

\begin{remark}\label{rem:nonempty}
Note that the intersection $\mathcal{U}^f_{\infty}(x) \cap \mathcal{U}(x)$ is ensured to be non-empty for all
$x \in \mathcal{X}_\infty$ since $\mathcal{X}_\infty$ is control invariant.
\end{remark}

In the next subsection, we show
that both $\mathcal{X}_\infty$ and $\mathcal{U}^f_\infty(x) \cap
\mathcal{U}(x)$ can be obtained \textit{simultaneously} by extending the
computation of invariant sets to the \textit{joint state-control space}.

\subsection{State-control sets}
\label{subsec:state-control-sets}

Let $z = \begin{bmatrix} x^\top & u^\top \end{bmatrix}^\top \in \mathbb{R}^{n_x + n_u}$ denote a state-control vector, and let $Z \subseteq \mathbb{R}^{n_x + n_u}$ be a set in the joint state-control space. We now define the two operations that allow us to decompose $Z$ into the state and control spaces, respectively.
\begin{definition}[State projection and $x$-section]\label{def:proj-section}
Let $Z \subseteq \mathbb{R}^{n_x + n_u}$ be a set in the joint state-control
space. The \emph{state projection} of $Z$ is the mapping
$\Pi_x:2^{\mathbb{R}^{n_x + n_u}} \to 2^{\mathbb{R}^{n_x}}$ and, for a given
$x \in \Pi_x(Z)$, the \emph{$x$-section} of $Z$ is the mapping
$\mathcal{U}_Z:\mathbb{R}^{n_x} \to 2^{\mathbb{R}^{n_u}}$, given by
\begin{equation*}
\begin{aligned}
    \Pi_x(Z) &\triangleq \big\{ x \in \mathbb{R}^{n_x} :
    \exists\, u \in \mathbb{R}^{n_u} \text{ s.t. }
    \begin{bmatrix} x^\top & u^\top \end{bmatrix}^\top\in Z\big\}, \\
    \mathcal{U}_Z(x) &\triangleq \big\{ u \in \mathbb{R}^{n_u} :
    \begin{bmatrix} x^\top & u^\top \end{bmatrix}^\top \in Z\big\}.
\end{aligned}
\end{equation*}
\end{definition}

Note that we do not use the control projection $\Pi_u(Z)$ since it would yield all
control inputs in $Z$ irrespective of the state. The $x$-section
$\mathcal{U}_Z(x)$ preserves the coupling between states and controls in
$Z$. Using these operations and given the state and input constraint sets $\mathcal{X}$ and $\mathcal{U}(x)$ for all $x \in \mathcal{X}$, we construct the joint constraint set in
the state-control space
\begin{equation}\label{eq:XtimesU}
    \mathcal{Z} \triangleq \big\{z \in \mathbb{R}^{n_x + n_u} : x \in \mathcal{X},\;
    u \in \mathcal{U}(x)\big\}.
\end{equation}
By construction, $\Pi_x(\mathcal{Z}) = \mathcal{X}$ and
$\mathcal{U}_\mathcal{Z}(x) = \mathcal{U}(x)$ for all $x \in \mathcal{X}$.

\subsection{State-control invariance}
\label{subsec:sci}
In this subsection, we introduce the concept of \textit{state-control invariant sets} and develop a recursive algorithm to compute these sets.

We now define the notion of \textit{state-control invariant sets}.
\begin{definition}[State-control invariant set]\label{def:sci}
Consider system \eqref{eq:dyn} and the joint constraint set $\mathcal{Z}$ given by 
\eqref{eq:XtimesU}. A set
$\mathcal{C} \subseteq \mathcal{Z}$ is called \emph{state-control
invariant} if, for every
$z \in \mathcal{C}$,
there exists $u^+ \in \mathbb{R}^{n_u}$ such that
$\begin{bmatrix} f(x,u)^\top & (u^+)^\top \end{bmatrix}^\top
\in \mathcal{C}$.

\end{definition}
Intuitively, Definition~\ref{def:sci} implies that for a given state-control invariant set $\mathcal{C}$, all state-control pairs $(x,u)$ in $\mathcal{C}$ yield next states $f(x,u) \in \Pi_x(\mathcal{C})$ for which the $x$-section $\mathcal{U}_\mathcal{C}(f(x,u))$ is non-empty. Using Definition~\ref{def:sci}, we now define the concept of \textit{maximal state-control invariant sets} (MSCI).

\begin{definition}[Maximal state-control invariant set]\label{def:msci}
A %state-control invariant
set $\mathcal{Z}_\infty \subseteq \mathcal{Z}$ is called the \emph{maximal
state-control invariant set} if
\begin{enumerate}[(i)]
    \item $\mathcal{Z}_\infty$ is state-control invariant, and
    \item if $\mathcal{C} \subseteq \mathcal{Z}$ is any state-control
    invariant set, then $\mathcal{C} \subseteq \mathcal{Z}_\infty$.
\end{enumerate}
\end{definition}
Next, we compute $\mathcal{Z}_\infty$ for which the state
projection yields the MCI set, that is,
$\Pi_x(\mathcal{Z}_\infty) = \mathcal{X}_\infty$, and its $x$-sections
yield the admissible invariance-preserving inputs, that is,
$\mathcal{U}_{\mathcal{Z}_\infty}(x) = \mathcal{U}^f_\infty(x)
\cap\, \mathcal{U}(x)$ for all $x \in \mathcal{X}_\infty$.
Similarly to the state-space case, $\mathcal{Z}_\infty$ can be computed
using a predecessor operator in the joint state-control space. For any set
$\Omega \subseteq \mathbb{R}^{n_x + n_u}$, define
\begin{multline}\label{eq:pre-z}
    \operatorname{Pre}_z(\Omega)
    \triangleq \Big\{
    z \in \mathbb{R}^{n_x + n_u} :\; \exists\, u^+ \in \mathbb{R}^{n_u} \\
    \text{s.t.}\;
    \begin{bmatrix}
    f(x,u)^\top & (u^+)^\top
    \end{bmatrix}^\top
    \in \Omega
    \Big\},
\end{multline}
that is, the set of all state-control vectors for which the next state $f(x,u)$ has a non-empty $x$-section, i.e., $\mathcal{U}_\Omega(f(x,u)) \neq \emptyset$.
The MSCI is then obtained as the fixed
point of the recursive iteration
\begin{equation}\label{eq:mci-z-iteration}
    \Omega_{k+1} = \operatorname{Pre}_z(\Omega_k) \cap \Omega_k,
    \quad \Omega_0 = \mathcal{Z}, \quad\forall k\in\overline{\mathbb{Z}}_{+}.
\end{equation}
At each step, $\Omega_{k+1}$ retains only those state-control vectors
in $\Omega_k$ for which an input exists that, when paired with the next state, lies
within $\Omega_k$. State-control vectors that lack such an input
are discarded. 

% Since the sequence $\{\Omega_k\}$ is nested by
% construction ($\Omega_{k+1} \subseteq \Omega_k$), the maximal
% state-control invariant set is defined as the fixed point of this iteration
% \begin{equation*}
%     \mathcal{Z}_\infty \triangleq \bigcap_{k=0}^{\infty} \Omega_k.
% \end{equation*}

Next, we show that under the same assumptions as in the
state-space case, the recursion \eqref{eq:mci-z-iteration} converges
to a non-empty limit set $\lim_{k\to \infty} \bigcap_{i=0}^k \Omega_i$, which is the MSCI, that is, $\mathcal{Z}_\infty = \lim_{k\to \infty} \bigcap_{i=0}^k \Omega_i$.

\begin{lemma}[Convergence of predecessor recursion in state-control space]
\label{lemma:z-convergence}
Consider system \eqref{eq:dyn} with joint constraint set $\mathcal{Z}$ given by \eqref{eq:XtimesU}.
If $\mathcal{Z}$ is compact and contains the
origin in its interior, and if $f$ is Lipschitz
continuous with $f(0,0) = 0$, then the recursion
\eqref{eq:mci-z-iteration} converges to a non-empty limit set $\lim_{k\to \infty} \bigcap_{i=0}^k \Omega_i$, and this limit is the MSCI satisfying $\mathcal{Z}_\infty = \lim_{k\to \infty} \bigcap_{i=0}^k \Omega_i$ contained in $\mathcal{Z}$.
\end{lemma}
\begin{proof}
By construction, $\{\Omega_k\}_{k \geq 0}$ is a nested sequence of
sets, i.e., $\Omega_{k+1} \subseteq \Omega_k$ for all $k \geq 0$.
Since $\Omega_0 = \mathcal{Z}$ is compact and each $\Omega_{k+1}$
is obtained as the intersection of $\Omega_k$ with the closed set
$\operatorname{Pre}_z(\Omega_k)$ (closedness follows from the
continuity of $f$), each $\Omega_k$ is compact. Furthermore, since
$f(0,0) = 0$ and the origin lies in the interior of $\mathcal{Z}$,
there exists a neighborhood of the origin that is state-control
invariant, ensuring that $\Omega_k$ is nonempty for all $k \geq 0$.
It remains to show that $\mathcal{Z}_\infty$ is state-control invariant.
Let $z = \begin{bmatrix} x^\top & u^\top \end{bmatrix}^\top \in
\mathcal{Z}_\infty$. Then $z \in \Omega_k$ for all $k \geq 0$, and
in particular $z \in \Omega_{k+1} =
\operatorname{Pre}_z(\Omega_k) \cap \Omega_k$ for all $k \geq 0$.
By the definition of $\operatorname{Pre}_z$, for each $k$ there exists
$u^+_k \in \mathbb{R}^{n_u}$ such that
$\begin{bmatrix} f(x,u)^\top & (u^+_k)^\top \end{bmatrix}^\top
\in \Omega_k$. Since each $\Omega_k$ is compact, the sequence
$\{u^+_k\}$ admits a convergent subsequence with limit $u^+ \in
\mathbb{R}^{n_u}$, and by the nested compactness of the $\Omega_k$,
$\begin{bmatrix} f(x,u)^\top & (u^+)^\top \end{bmatrix}^\top
\in \mathcal{Z}_\infty$. Hence $\mathcal{Z}_\infty$ is state-control
invariant.
Maximality follows by the same argument as in the state-space
case~\cite{borrelli2017predictive}: if $\mathcal{C} \subseteq
\mathcal{Z}$ is any state-control invariant set, then by induction
$\mathcal{C} \subseteq \Omega_k$ for all $k \geq 0$, and hence
$\mathcal{C} \subseteq \mathcal{Z}_\infty$.
By the finite intersection property for compact
sets,
$\mathcal{Z}_\infty = \bigcap_{k=0}^{\infty} \Omega_k$
is nonempty and compact.\frQED
\end{proof}

The next lemma states that the projection of the MSCI $\mathcal{Z}_\infty$ is the MCI set $\mathcal{X}_\infty$ and the $x$-sections of $\mathcal{Z}_\infty$, for all $x \in \mathcal{X}_\infty$, are the admissible invariance-preserving inputs of $\mathcal{X}_\infty$.
\begin{lemma}[MCI and admissible invariance-preserving inputs]
\label{lemma:z-recovery}
Consider system \eqref{eq:dyn} with joint constraint set $\mathcal{Z}$ given by \eqref{eq:XtimesU}. The MSCI $\mathcal{Z}_\infty$ admits the following properties:
\begin{enumerate}[(i)]
    \item $\Pi_x(\mathcal{Z}_\infty) = \mathcal{X}_\infty$, and
    \item $\mathcal{U}_{\mathcal{Z}_\infty}(x) =
    \mathcal{U}^f_\infty(x) \cap\, \mathcal{U}(x)$
    for all $x \in \mathcal{X}_\infty$.
\end{enumerate}
\end{lemma}

\begin{proof}
    We prove the claims separately.
 
\emph{Proof of (i).} We show the equality by double inclusion.
 
($\subseteq$) Let $x \in \Pi_x(\mathcal{Z}_\infty)$. Then there exists
$u \in \mathbb{R}^{n_u}$ such that
$z = \begin{bmatrix} x^\top & u^\top \end{bmatrix}^\top
\in \mathcal{Z}_\infty \subseteq \mathcal{Z}$,
so $x \in \mathcal{X}$ and $u \in \mathcal{U}(x)$. Since
$\mathcal{Z}_\infty$ is state-control invariant, there exists
$u^+ \in \mathbb{R}^{n_u}$ such that
$\begin{bmatrix} f(x,u)^\top & (u^+)^\top \end{bmatrix}^\top
\in \mathcal{Z}_\infty$. In particular,
$f(x,u) \in \Pi_x(\mathcal{Z}_\infty)$. By induction, the state can be kept in
$\Pi_x(\mathcal{Z}_\infty) \subseteq \mathcal{X}$ for all time via
admissible inputs (i.e., by applying $u^+$ at $x^+ = f(x,u)$), so $\Pi_x(\mathcal{Z}_\infty)$ is control invariant.
By maximality of $\mathcal{X}_\infty$, we have
$\Pi_x(\mathcal{Z}_\infty) \subseteq \mathcal{X}_\infty$.
 
($\supseteq$) Let $x \in \mathcal{X}_\infty$. Since
$\mathcal{X}_\infty$ is control invariant, there exists
$u \in \mathcal{U}(x)$ such that $f(x,u) \in \mathcal{X}_\infty$.
Define the set
$\mathcal{C} = \big\{ \begin{bmatrix} x^\top & u^\top
\end{bmatrix}^\top : x \in \mathcal{X}_\infty,\;
u \in \mathcal{U}(x),\; f(x,u) \in \mathcal{X}_\infty \big\}$.
Then $\mathcal{C} \subseteq \mathcal{Z}$ and $\mathcal{C}$ is
state-control invariant: for any
$z \in \mathcal{C}$, the next state
$x^+ = f(x,u) \in \mathcal{X}_\infty$, so by control invariance of
$\mathcal{X}_\infty$ there exists $u^+ \in \mathcal{U}(x^+)$ with
$f(x^+, u^+) \in \mathcal{X}_\infty$, yielding
$\begin{bmatrix} (x^+)^\top & (u^+)^\top \end{bmatrix}^\top
\in \mathcal{C}$. By maximality of $\mathcal{Z}_\infty$,
$\mathcal{C} \subseteq \mathcal{Z}_\infty$, and hence
$x \in \Pi_x(\mathcal{Z}_\infty)$.
 
\emph{Proof of (ii).} Let $x \in \mathcal{X}_\infty$.
 
($\subseteq$) Let $u \in \mathcal{U}_{\mathcal{Z}_\infty}(x)$. Then
$\begin{bmatrix} x^\top & u^\top \end{bmatrix}^\top
\in \mathcal{Z}_\infty \subseteq \mathcal{Z}$, so
$u \in \mathcal{U}(x)$. Since $\mathcal{Z}_\infty$ is state-control
invariant, there exists $u^+ \in \mathbb{R}^{n_u}$ such that
$\begin{bmatrix} f(x,u)^\top & (u^+)^\top \end{bmatrix}^\top
\in \mathcal{Z}_\infty$. In particular,
$f(x,u) \in \Pi_x(\mathcal{Z}_\infty) = \mathcal{X}_\infty$ by (i).
Hence $u \in \mathcal{U}^f_\infty(x)$ by the definition of
$\mathcal{U}^f_\infty$, and so
$u \in \mathcal{U}^f_\infty(x) \cap\, \mathcal{U}(x)$.
 
($\supseteq$) Let
$u \in \mathcal{U}^f_\infty(x) \cap\, \mathcal{U}(x)$.
Then $u \in \mathcal{U}(x)$ and
$f(x,u) \in \mathcal{X}_\infty$. Since
$x \in \mathcal{X}_\infty$ and $u \in \mathcal{U}(x)$ and
$f(x,u) \in \mathcal{X}_\infty$, we have
$\begin{bmatrix} x^\top & u^\top \end{bmatrix}^\top \in \mathcal{C}
\subseteq \mathcal{Z}_\infty$, where $\mathcal{C}$ is as constructed
in the proof of (i). Hence
$u \in \mathcal{U}_{\mathcal{Z}_\infty}(x)$. \frQED
\end{proof}

%%%%%%%%%%%%%%%%%%%%%%%%%%%%%%%%%%%%%%%%%%%%%%%%%%%%%%%%%%%%%%%%%%%%%%%%%%%%%%%%
\section{Learning from Failures}
\label{sec:learning}
In this section, we develop an algorithm to learn the MSCI from failing trajectories for deterministic LTI systems with unknown parameters and polytopic constraints. We specialize to LTI systems with polytopic constraints because each predecessor constraint defines a halfspace whose normal can be identified from failures in closed form.

\subsection{Specialization to linear systems with polytopic constraints}
\label{subsec:lti-polytopic}

For the remainder of this paper, we consider the deterministic LTI
system
\begin{equation}\label{eq:lti}
    x(k+1) = Ax(k) + Bu(k), \quad x(0) = x_0, \quad\forall k\in\overline{\mathbb{Z}}_{+},
\end{equation}
where $A \in \mathbb{R}^{n_x \times n_x}$ and
$B \in \mathbb{R}^{n_x \times n_u}$. Consider the polytopic constraints
in the joint state-control space given in
$\mathcal{H}$-representation~\cite[Sec.~4.2]{borrelli2017predictive} as
\begin{equation}
\label{eq:P-def}
\mathcal{P}
=
\{z\in\mathbb{R}^{n_x+n_u}:\ H_z\, z\le h_z\},
\end{equation}
where $H_z\in\mathbb{R}^{n_\mathrm{c}\times(n_x+n_u)}$ and
$h_z\in\mathbb{R}^{n_\mathrm{c}}$, with $n_\mathrm{c} \in \mathbb{Z}^+$ denoting the
number of constraints. We partition the constraint matrix $H_z$ as
\begin{equation*}
H_z
=
\begin{bmatrix}
H_{zx} & H_{zu}
\end{bmatrix},
\end{equation*}
where $H_{zx} \in \mathbb{R}^{n_\mathrm{c} \times n_x}$ and
$H_{zu} \in \mathbb{R}^{n_\mathrm{c} \times n_u}$, so that 
\[
(H_{zx})_j\, x + (H_{zu})_j\, u \le (h_z)_j, \qquad \forall j \in \{1, \dots, n_\mathrm{c}\}.
\]
For any $x \in \Pi_x(\mathcal{P})$, the $x$-section of $\mathcal{P}$ takes the form
\begin{equation}
\label{eq:Ux-from-partition}
\mathcal{U}_{\mathcal{P}}(x)
=
\big\{u\in\mathbb{R}^{n_u}:\
H_{zu}\, u \le h_z - H_{zx}\, x\big\},
\end{equation}
which is a polytope in $\mathbb{R}^{n_u}$ parameterized by $x$.

The predecessor of a polytope $\Omega \subset \mathbb{R}^{n_x+n_u}$ is given by $
\operatorname{Pre}_z(\Omega)
    = \big\{z \in \mathbb{R}^{n_x+n_u} : \exists u^+ \text{ s.t. }
    \begin{bmatrix} (Ax+Bu)^\top & (u^+)^\top \end{bmatrix}^\top
\in \Omega\big\}
$, which can alternatively be written as
\begin{equation}\label{eq:pre-z-lti}
    \operatorname{Pre}_z(\Omega)
    = \big\{z \in \mathbb{R}^{n_x+n_u} :
    \begin{bmatrix} A & B \end{bmatrix} z \in \Pi_x(\Omega)\big\}.
\end{equation}
Similarly to computing the MCI \cite{borrelli2017predictive}, we can apply the recursion \eqref{eq:mci-z-iteration} to obtain the MSCI in a finite number of iterations.

Since $\Omega$ is a polytope, its projection $\Pi_x(\Omega)$ is a polytope~\cite[Prop.~15]{jones2004equality} and hence admits an $\mathcal{H}$-representation given by 
\begin{equation}
    \label{eq:proj-poly}
    \Pi_x(\Omega) = \{x \in \mathbb{R}^{n_x} :H_{\mathrm{proj}}\, x \leq g_{\mathrm{proj}}\}
\end{equation} 
where $H_{\mathrm{proj}} \in \mathbb{R}^{n_\mathrm{p}\times n_x}$ and $g_{\mathrm{proj}} \in \mathbb{R}^{n_\mathrm{p}}$, with $n_\mathrm{p} \in \mathbb{Z}^+$ denoting the number of halfspaces of the projected polytope. Note that $Ax+Bu \in \Pi_x(\Omega)$ if and only if
\begin{equation*}
    \label{eq:next-in-proj}
    H_{\mathrm{proj}}\, (Ax+Bu) \leq g_{\mathrm{proj}},
\end{equation*}
which, noting $Ax+Bu = \begin{bmatrix} A & B \end{bmatrix} z$ and using \eqref{eq:proj-poly}, yields the $\mathcal{H}$-representation of $\operatorname{Pre}_z(\Omega)$ as
\begin{equation}\label{eq:pre-z-lti-explicit}
    \operatorname{Pre}_z(\Omega)
    = \big\{z \in \mathbb{R}^{n_x+n_u} :
    H_{\mathrm{proj}}
    \begin{bmatrix} A & B \end{bmatrix} z
    \leq g_{\mathrm{proj}}\big\},
\end{equation}
which is a polytope in the state-control space with $n_\mathrm{p}$ constraints. Hence, there is a row-wise correspondence between the constraints defined by the projected polytope $\Pi_x(\Omega)$ and those of the predecessor $\operatorname{Pre}_z(\Omega)$.
\subsection{Iterative constraint refinement}
\label{subsec:iterative-refinement}
In this subsection, we develop an iterative learning algorithm for learning the MSCI $\mathcal{Z}_\infty$. Rather than computing $\mathcal{Z}_\infty$ via predecessor recursion, which requires knowledge of $(A,B)$, we iteratively refine a polytopic constraint set at each iteration using observed failures.

Let the constraint set at iteration $\ell \in \overline{\mathbb{Z}}_{+}$ be
\begin{equation*}
    \mathcal{P}_\ell
    \triangleq
    \{z \in \mathbb{R}^{n_x+n_u} :\
    H_z^{\,\ell}\, z \le h_z^{\,\ell}\},
\end{equation*}
with $H_z^{\,\ell} \in \mathbb{R}^{n_\mathrm{c}^\ell \times (n_x+n_u)}$,
$h_z^{\,\ell} \in \mathbb{R}^{n_\mathrm{c}^\ell}$, where $n_\mathrm{c}^\ell \in \mathbb{Z}^+$ is the
number of constraints at the $\ell$th iteration. Let the initial polytope $\mathcal{P}_0$ be $\left\{z:x\in\mathcal{X}, u\in\mathcal{U}(x)\right\}$. Using the state projection and $x$-section of $\mathcal{P}_\ell$, 
the state and input constraints at iteration $\ell$ are
\begin{equation*}
    \mathcal{X}_\ell \triangleq \Pi_x(\mathcal{P}_\ell),
    \qquad
    \mathcal{U}_\ell(x) \triangleq \mathcal{U}_{\mathcal{P}_\ell}(x), \qquad \forall x\in\mathcal{X},
\end{equation*}
so that $\mathcal{U}_\ell(x) = \{u \in \mathbb{R}^{n_u} :
H_{zu}^{\,\ell}\, u \le h_z^{\,\ell} - H_{zx}^{\,\ell}\, x\}$
for all $x \in \mathcal{X}_\ell$. The $x$-section $\mathcal{U}_\ell(x)$ is the algorithm's current
estimate of the admissible invariance-preserving inputs at state $x \in \mathcal{X}_\ell$. The $\mathcal{H}$-representation of $\mathcal{X}_{\ell}$ is
\begin{equation*}
    \mathcal{X}_{\ell} = \{x \in \mathbb{R}^{n_x} : H_{\mathrm{proj}}^{\,\ell}\, x \leq g_{\mathrm{proj}}^{\,\ell}\},
\end{equation*}
where $H_{\mathrm{proj}}^{\,\ell} \in \mathbb{R}^{n_\mathrm{p}^{\ell} \times n_x}$ and $g_{\mathrm{proj}}^{\,\ell} \in \mathbb{R}^{n_\mathrm{p}^{\ell}}$ where $n_\mathrm{p}^\ell \in \mathbb{Z}^+$ is the
number of constraints defined by $\mathcal{X}_\ell$ at the $\ell$th iteration.

At iteration $\ell \in \mathbb{Z}^+$, a state-feedback controller
$\pi^{\,\ell} : \mathcal{X}_{\ell-1} \to \mathbb{R}^{n_u}$
satisfying
$\pi^{\,\ell}(x) \in \mathcal{U}_{\ell-1}(x)$ for all
$x \in \mathcal{X}_{\ell-1}$ is applied to system \eqref{eq:lti},
generating a closed-loop trajectory defined by $\mathcal{T}^\ell \triangleq \{z^{\ell}(k)\}_{k\geq0}$ with
$z^{\ell}(k) = \begin{bmatrix} (x^{\ell}(k))^\top &
(u^{\ell}(k))^\top \end{bmatrix}^\top$ and
$u^{\ell}(k) = \pi^{\,\ell}(x^{\ell}(k))$ for all $k \geq 0$. We assume that at every iteration $\ell$ the controller $\pi^\ell$ generates a failing trajectory.

\begin{assumption}[Permissible failure]
\label{assume:permissible-failure}
At each iteration $\ell$, $\pi^\ell$ is permitted to generate a failing trajectory. Upon failure, the system is reset to an initial state $x(0) \in \mathcal{X}_{\ell-1}$.
\end{assumption}

Assumption~\ref{assume:permissible-failure} restricts FAIL to systems on which a constraint violation is recoverable, such as simulation models, hardware testbeds on which 
$\mathcal{X}$ encodes an operational envelope, and systems whose logged data already contains violations.

At iteration $\ell$, a one-step failing state-input vector $z^\ell = \begin{bmatrix}
    (x^{\ell})^\top & (u^{\ell})^\top
\end{bmatrix}^\top$, with $x^\ell \in \mathcal{X}_{\ell-1}$ and $u^{\ell}\in \mathcal{U}_{\ell-1}(x^\ell)$, satisfies
\begin{equation}
\label{eq:failure-condition}
    z^{\ell} \in \mathcal{P}_{\ell-1}
    \qquad \text{and} \qquad
    z^{\ell} \notin
    \operatorname{Pre}_z(\mathcal{P}_{\ell-1}),
\end{equation}
which implies that either $u^{\ell}$ is \emph{not} in the invariance-preserving input set $\mathcal{U}^f_\infty(x^{\ell})$ at $x^\ell$, or $x^\ell$ is \textit{not} in $\mathcal{X}_\infty$. Our goal is to exclude such states and inputs from $\mathcal{P}_\ell$, the polytope used in the next iteration.

Now, we collect the failure time index and the corresponding one-step failing state-control pair from the trajectory at iteration $\ell$ as
\begin{equation}
\label{eq:failure-data}
\begin{aligned}
    &\mathcal{V}^{\,\ell}
    \triangleq \big\{k : z^{\ell}(k) \in \mathcal{P}_{\ell-1},\;
    x^{\ell}(k{+}1) \notin \mathcal{X}_{\ell-1}\big\},
    \\
    & \mathcal{F}^{\,\ell}
    \triangleq \big\{z^{\ell}(k) : k \in \mathcal{V}^{\,\ell}\big\}.
\end{aligned}
\end{equation}

Since $\operatorname{Pre}_z(\mathcal{P}_{\ell-1})$ is a polytope, each $z \in \mathcal{F}^{\,\ell}$ violates at least one of the $n_\mathrm{p}^{\ell-1}$ constraints associated with
$\operatorname{Pre}_z(\mathcal{P}_{\ell-1})$. 

By \eqref{eq:pre-z-lti-explicit}, the $j$th constraint of $\mathcal{X}_{\ell-1}$ induces the $j$th constraint of $\operatorname{Pre}_z(\mathcal{P}_{\ell-1})$ given by
\begin{equation}
    \label{eq:violated_halfspace}
    (H_{\mathrm{proj}}^{\,\ell-1})_j \begin{bmatrix} A & B \end{bmatrix}
    z \leq (g_{\mathrm{proj}}^{\,\ell-1})_j.
\end{equation}
Define the vector $a_j^* \triangleq \begin{bmatrix} A & B \end{bmatrix}^\top (H_{\mathrm{proj}}^{\,\ell-1})^\top_j
\in \mathbb{R}^{n_x+n_u}$ as the $z$-space normal of the halfspace that we seek to
learn.

When a failure occurs at time $k$, the observed next state
$x^{\ell}(k{+}1)$ violates at least one constraint of
$\mathcal{X}_{\ell-1}$ which is associated with a halfspace in $\operatorname{Pre}_z(\mathcal{P}_{\ell-1})$. Examining which constraints defined by $\mathcal{X}_{\ell-1}$ were violated allows us to identify \emph{which} predecessor constraints were violated and their corresponding $(g_{\mathrm{proj}}^{\,\ell-1})_j$ without knowing the dynamics. 

Next, we learn $a_j^*$ without knowing the dynamics. The $z$-space normal $a_j^*$ satisfies
\begin{equation}
\label{eq:regression-relation}
\begin{aligned}
    a_j^{*\top} z(t)
    &= (H_{\mathrm{proj}}^{\,\ell-1})_j\, \begin{bmatrix} A & B \end{bmatrix} z(t) \\
    &= (H_{\mathrm{proj}}^{\,\ell-1})_j\, x(t{+}1),
\end{aligned}
\end{equation}
for every time step $t$ along any trajectory of \eqref{eq:lti}. The left-hand side involves the unknown $a_j^*$, while the right-hand side is computed using $(H_{\mathrm{proj}}^{\,\ell-1})_j$ and $x(t{+}1)$. Hence, learning $a_j^*$ is a linear regression with $z(t)$ being the regressor in $\mathbb{R}^{n_x + n_u}$, $a_j^*$ being the parameter to estimate, and $(H_{\mathrm{proj}}^{\,\ell-1})_j\, x(t{+}1)$ being a measurement. 

Let
$p \triangleq n_x + n_u$. To learn $a_j^*$, we require $p$ linearly independent state-control samples. We form a window of $p$ samples from available trajectory data
\begin{equation}
\label{eq:Z-window}
    Z =
    \begin{bmatrix}
    z(t_1)^\top \\
    z(t_2)^\top \\
    \vdots \\
    z(t_p)^\top
    \end{bmatrix}
    \in \mathbb{R}^{p \times p},
    \;
    s =
    \begin{bmatrix}
    (H_{\mathrm{proj}}^{\,\ell-1})_j\, x(t_1{+}1) \\
    (H_{\mathrm{proj}}^{\,\ell-1})_j\, x(t_2{+}1) \\
    \vdots \\
    (H_{\mathrm{proj}}^{\,\ell-1})_j\, x(t_p{+}1)
    \end{bmatrix}
    \in \mathbb{R}^{p},
\end{equation}
where $t_1, \ldots, t_p$ are time indices drawn from the failing trajectory or from other available trajectories.
Since \eqref{eq:lti} is linear and deterministic, $s = Z a_j^*$ holds. If
$\operatorname{rank}(Z) = p$, the estimate of the normal vector $\widehat{a}_j$ is given by
\begin{equation}
\label{eq:normal-lstsq}
    \widehat{a}_j= a_j^* = Z^{-1} s.
\end{equation}

Note that the rank condition $\operatorname{rank}(Z) = p$ is a \textit{persistence of excitation} condition on the state-control data. 

\begin{remark}[Learning the dynamics]\label{rem:compare-model-learning}
    The collected data $Z$ could be used to learn the dynamics $(A,B)$. However, this involves learning $n_x \times (n_x+n_u)$ parameters as opposed to $n_x + n_u$ parameters per halfspace in our method. Furthermore, recursion~\eqref{eq:mci-z-iteration} would still need to be run after $(A,B)$ are learned to compute $\mathcal{Z}_\infty$.
\end{remark}

We now prove that our learning algorithm does not eliminate any state-control pairs lying in $\mathcal{Z}_\infty$.
\begin{lemma}[Exact recovery of predecessor constraint]
\label{lemma:exact-recovery}
Given an iteration index $\ell \in \mathbb{Z}^+$, consider system~\eqref{eq:lti} with state and input constraints $\mathcal{X}_{\ell-1}$ and $\mathcal{U}_{\ell-1}(x)$ for all $x \in \mathbb{R}^{n_x}$. Suppose the $j$th constraint of $\mathcal{X}_{\ell-1} = \Pi_x(\mathcal{P}_{\ell-1})$ is violated by $\mathcal{T}^\ell$, and let $a_j^* = \begin{bmatrix} A & B \end{bmatrix}^\top (H_{\mathrm{proj}}^{\,\ell-1})_j^\top$. If $\operatorname{rank}(Z) = p$, then the halfspace $\{z \in \mathbb{R}^{n_x+n_u} : (Z^{-1}s)^\top z \leq (g_{\mathrm{proj}}^{\,\ell-1})_j\}$ is one of the $n_\mathrm{p}^{\ell-1}$ constraints defining $\operatorname{Pre}_z(\mathcal{P}_{\ell-1})$, and $\forall z \in \mathcal{Z}_\infty$ we have $(Z^{-1}s)^\top z \leq (g_{\mathrm{proj}}^{\,\ell-1})_j$.
\end{lemma}

\begin{proof}
Since the dynamics \eqref{eq:lti} are linear,
$(H_{\mathrm{proj}}^{\,\ell-1})_j\, x(t{+}1) = (H_{\mathrm{proj}}^{\,\ell-1})_j\, (Ax(t) +
Bu(t)) = a_j^{*\top} z(t)$ for every $t$. Therefore, $s = Z a_j^*$ holds. If $\operatorname{rank}(Z) = p$, then $Z$ is invertible
and $\widehat{a}_j = Z^{-1}s = a_j^*$ by \eqref{eq:normal-lstsq}. By \eqref{eq:pre-z-lti-explicit}, the halfspace
$a_j^{*\top} z \leq (g_{\mathrm{proj}}^{\,\ell-1})_j$ is one of the $n_\mathrm{p}^{\ell-1}$ constraints defining
$\operatorname{Pre}_z(\mathcal{P}_{\ell-1})$, establishing the first
claim.

For the second claim, note that $\mathcal{Z}_\infty$ is
state-control invariant and
$\mathcal{Z}_\infty \subseteq \mathcal{P}_{\ell-1}$ since
$\mathcal{Z}_\infty \subseteq \mathcal{P}_0$ and the sequence
$\{\mathcal{P}_\ell\}$ is non-increasing. For any
$z = \begin{bmatrix} x^\top & u^\top \end{bmatrix}^\top \in
\mathcal{Z}_\infty$, the next state satisfies $Ax + Bu \in \Pi_x(\mathcal{Z}_\infty) \subseteq \mathcal{X}_{\ell-1}$ by the state-control invariance of $\mathcal{Z}_\infty$ and the nesting $\mathcal{X}_\infty \subseteq \mathcal{X}_{\ell-1}$. Hence
$(H_{\mathrm{proj}}^{\,\ell-1})_j\, (Ax + Bu) \leq (g_{\mathrm{proj}}^{\,\ell-1})_j$, i.e.,
$a_j^{*\top} z \leq (g_{\mathrm{proj}}^{\,\ell-1})_j$. No state-control pair in
$\mathcal{Z}_\infty$ is excluded. \frQED
\end{proof}

Next, we prove that the learned halfspace $\widehat{a}_j^\top z \leq \left(g^{\ell -1}_{\mathrm{proj}}\right)_j$ excludes the one-step failing state-input pair so that the same failure is not repeated in the next iteration.
\begin{lemma}[Failure exclusion]
\label{lemma:failure-exclusion}
Given an iteration index $\ell \in \mathbb{Z}^+$, consider system~\eqref{eq:lti} with state and input constraints $\mathcal{X}_{\ell-1}$ and $\mathcal{U}_{\ell-1}(x)$ for all $x \in \mathbb{R}^{n_x}$. Let $k$ be a failure time index, and $(H_{\mathrm{proj}}^{\,\ell-1})_j\, x \leq (g_{\mathrm{proj}}^{\,\ell-1})_j$ be the
corresponding violated constraint of $\mathcal{X}_{\ell-1}$, and let
$\widehat{a}_j = a_j^*$ be learned as in
Lemma~\ref{lemma:exact-recovery}. Then
$\widehat{a}_j^\top z^{\ell}(k) > (g_{\mathrm{proj}}^{\,\ell-1})_j$, and $z^{\ell}(k) \notin \mathcal{P}_\ell$ after the
halfspace $\widehat{a}_j^\top z \leq (g_{\mathrm{proj}}^{\,\ell-1})_j$ is added
to $\mathcal{P}_{\ell-1}$.
\end{lemma}
\begin{proof}
By \eqref{eq:regression-relation},
$\widehat{a}_j^\top z^{\ell}(k) = a_j^{*\top} z^{\ell}(k) =
(H_{\mathrm{proj}}^{\,\ell-1})_j\, x^{\ell}(k{+}1)$. Since $k$ is a failure
index, then $(H_{\mathrm{proj}}^{\,\ell-1})_j\, x^{\ell}(k{+}1) > (g_{\mathrm{proj}}^{\,\ell-1})_j$, and
hence $\widehat{a}_j^\top z^{\ell}(k) > (g_{\mathrm{proj}}^{\,\ell-1})_j$. Since
$\mathcal{P}_\ell \subseteq \{z : \widehat{a}_j^\top z \leq
(g_{\mathrm{proj}}^{\,\ell-1})_j\}$ by construction, we have
$z^{\ell}(k) \notin \mathcal{P}_\ell$. \frQED
\end{proof}
 
At each iteration $\ell$, for each failure point in $\mathcal{F}^{\ell}$, a halfspace $(\widehat{a}_j,\, (g_{\mathrm{proj}}^{\,\ell-1})_j)$ is learned if $\operatorname{rank}(Z) = p$. The polytope $\mathcal{P}_{\ell-1}$ is intersected with the learned halfspace to obtain the polytope used in the next iteration:
\begin{equation}
\label{eq:P-update}
    \mathcal{P}_{\ell}
    =
    \mathcal{P}_{\ell-1}
    \;\cap\;
    \{z \in \mathbb{R}^{n_x+n_u} :
    \widehat{a}_j^\top z \leq (g_{\mathrm{proj}}^{\,\ell-1})_j\}.
\end{equation}
Failure is checked using the updated polytope until no new failures are identified. The algorithm terminates when a maximum number of iterations $L$ is reached.

Now we analyze the convergence of the algorithm to $\mathcal{Z}_\infty$.
\begin{algorithm}[!b]
\caption{FAIL: Failure-Aware Iterative Learning}
\label{alg:fail-learn}
\begin{algorithmic}[1]
\Require Initial polytope $\mathcal{P}_0$, maximum number of iterations $L$, maximum trajectory length $T$.
\State $\ell \gets 1$
\While{$\ell \leq L$}
    \State Apply controller $\pi^\ell$ with
    $\pi^\ell(x) \in \mathcal{U}_{\ell-1}(x)$
    \State Compute $\mathcal{X}_{\ell-1} =
    \Pi_x(\mathcal{P}_{\ell-1})$
    \State Collect $\{z(k)\}_{k=0}^{T}$, terminate when $x(k)\notin \mathcal{X}_{\ell-1}$
    \While{\textbf{true}}
        \State $\ell' \gets \ell$
        \State Collect $\mathcal{F}^{\ell'}$ from $\{z(k)\}$ using \eqref{eq:failure-data} with $\mathcal{X}_{\ell-1}$
        \If{$\mathcal{F}^{\ell'}$ is $\emptyset$}
            \State \textbf{break} \Comment{no new failures from this trajectory}
        \EndIf
        \For{each $z(k) \in \mathcal{F}^{\ell'}$}
            \State Identify failure $(H_{\mathrm{proj}}^{\,\ell-1})_j\, x > (g_{\mathrm{proj}}^{\,\ell-1})_j$
            \State Form $Z$, $s$ from \eqref{eq:Z-window} using $p$ samples
            \If{$\operatorname{rank}(Z) \neq p$}
                \State Append data points until $\operatorname{rank}(Z) = p$
            \EndIf
            \State $\widehat{a}_j \gets Z^{-1} s$
            \State $\mathcal{P}_\ell \gets \mathcal{P}_{\ell-1} \cap
            \{z : \widehat{a}_j^\top z \leq (g_{\mathrm{proj}}^{\,\ell-1})_j\}$
            \State $\mathcal{X}_{\ell} \gets \Pi_x(\mathcal{P}_{\ell})$
            \State $\ell \gets \ell + 1$
        \EndFor
    \EndWhile
\EndWhile
\end{algorithmic}
\end{algorithm}

\begin{theorem}[Monotone convergence]
\label{theorem:monotone-convergence}
Let $\{\mathcal{P}_\ell\}_{\ell \geq 0}$ be the sequence of
polytopes generated by \eqref{eq:P-update}. Then, the following statements hold:
\begin{enumerate}[(i)]
    \item $\mathcal{Z}_\infty \subseteq \mathcal{P}_\ell \subseteq
    \mathcal{P}_{\ell-1}$ for all $\ell \in \mathbb{Z}^+$.
    \item $\mathcal{U}^f_\infty(x) \cap \mathcal{U}(x) \subseteq
    \mathcal{U}_\ell(x) \subseteq \mathcal{U}_{\ell-1}(x)$ for all
    $x \in \mathcal{X}_\ell$ and all $\ell \in \mathbb{Z}^+$.
    \item Suppose that for every $\ell \in \overline{\mathbb{Z}}_{+}$
    with $\mathcal{P}_\ell \neq \mathcal{Z}_\infty$, the conditions
    of Lemma~\ref{lemma:exact-recovery} hold and the learned
    halfspace is not already a constraint of $\mathcal{P}_\ell$.
    Then $\mathcal{P}_\ell \to \mathcal{Z}_\infty$ in a finite
    number of iterations.
\end{enumerate}
\end{theorem}

\begin{proof}
\emph{(i)} The right inclusion
$\mathcal{P}_\ell \subseteq \mathcal{P}_{\ell-1}$ holds by
construction, since $\mathcal{P}_\ell$ is obtained by intersecting
$\mathcal{P}_{\ell-1}$ with additional halfspaces. The inclusion $\mathcal{Z}_\infty \subseteq \mathcal{P}_0$ holds since $\mathcal{Z}_\infty \subseteq \mathcal{Z} = \mathcal{P}_0$. Suppose $\mathcal{Z}_\infty \subseteq \mathcal{P}_{\ell-1}$. By Lemma~\ref{lemma:exact-recovery}, every $z \in \mathcal{Z}_\infty$ satisfies each learned halfspace $\widehat{a}^\top z \leq g$. Hence $\mathcal{Z}_\infty \subseteq \mathcal{P}_\ell$ holds by induction.

\emph{(ii)} For any $x \in \mathcal{X}_\ell$, the inclusion $\mathcal{U}_\ell(x) \subseteq \mathcal{U}_{\ell-1}(x)$ follows from $\mathcal{P}_\ell \subseteq \mathcal{P}_{\ell-1}$ and the definition of the $x$-section. For any $x \in \mathcal{X}_\ell$ and $u \in \mathcal{U}^f_\infty(x) \cap \mathcal{U}(x) =
\mathcal{U}_{\mathcal{Z}_\infty}(x)$, we have $\begin{bmatrix} x^\top & u^\top \end{bmatrix}^\top \in \mathcal{Z}_\infty \subseteq \mathcal{P}_\ell$, so
$u \in \mathcal{U}_{\mathcal{P}_\ell}(x) = \mathcal{U}_\ell(x)$, then $\mathcal{U}^f_\infty(x) \cap \mathcal{U}(x) \subseteq
    \mathcal{U}_\ell(x) \subseteq \mathcal{U}_{\ell-1}(x)$ holds by induction.

\emph{(iii)} By (i), the sequence
$\{\mathcal{P}_\ell\}$ is non-increasing and bounded below by
$\mathcal{Z}_\infty$. Since both $\mathcal{P}_\ell$ and $\mathcal{Z}_\infty$ are polytopes, and $\mathcal{Z}_\infty$ is obtained from $\mathcal{P}_0$ by adjoining a finite number of predecessor constraints, the total number of constraints defined by $\operatorname{Pre}_z(\mathcal{P}_\ell)$ not yet in $\mathcal{P}_\ell$ is finite and non-increasing across iterations. By assumption, at least one such constraint is learned at each iteration for which $\mathcal{P}_\ell \neq \mathcal{Z}_\infty$. Hence the process terminates with $\mathcal{P}_\ell = \mathcal{Z}_\infty$ in a finite number of iterations.\frQED
\end{proof}

Note that (i) and (ii) hold at every iteration, so any number of observed failures yields an outer approximation of $\mathcal{Z}_\infty$ that excludes no invariance-preserving pair, whereas (iii) requires that new predecessor constraints continue to be violated.
 
The Failure-Aware Iterative Learning (FAIL) algorithm is summarized in Algorithm~\ref{alg:fail-learn}.

\begin{figure}[!tb]
  \centering
  \begin{subfigure}{0.49\linewidth}
    \centering
    \includegraphics[trim=3.5cm 0.8cm 3.5cm 0.2cm, clip, width=\textwidth]{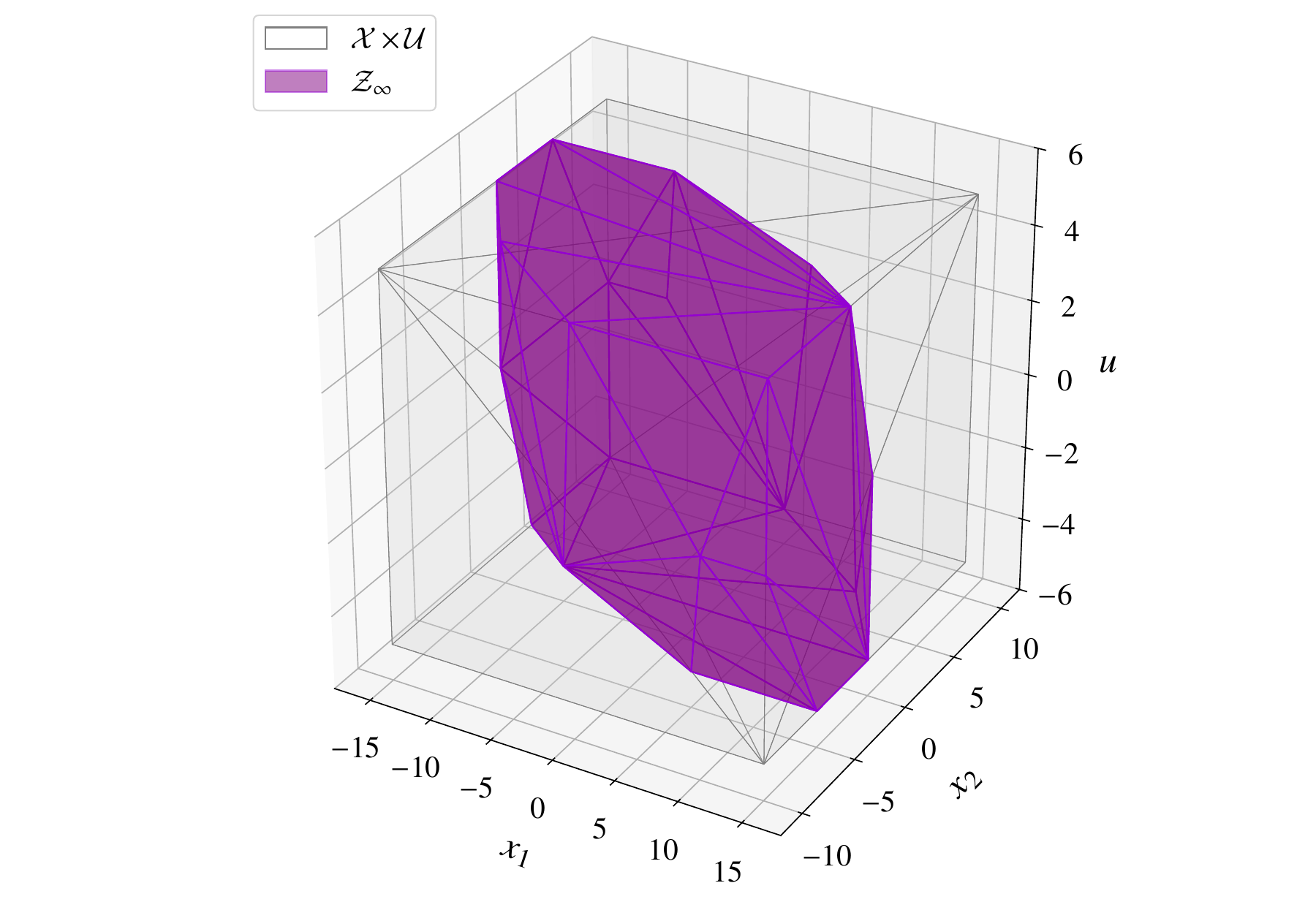}
    \caption{}\label{img:mci-z}
  \end{subfigure}\hfill
  \begin{subfigure}{0.49\linewidth}
    \centering
    \includegraphics[width=\textwidth]{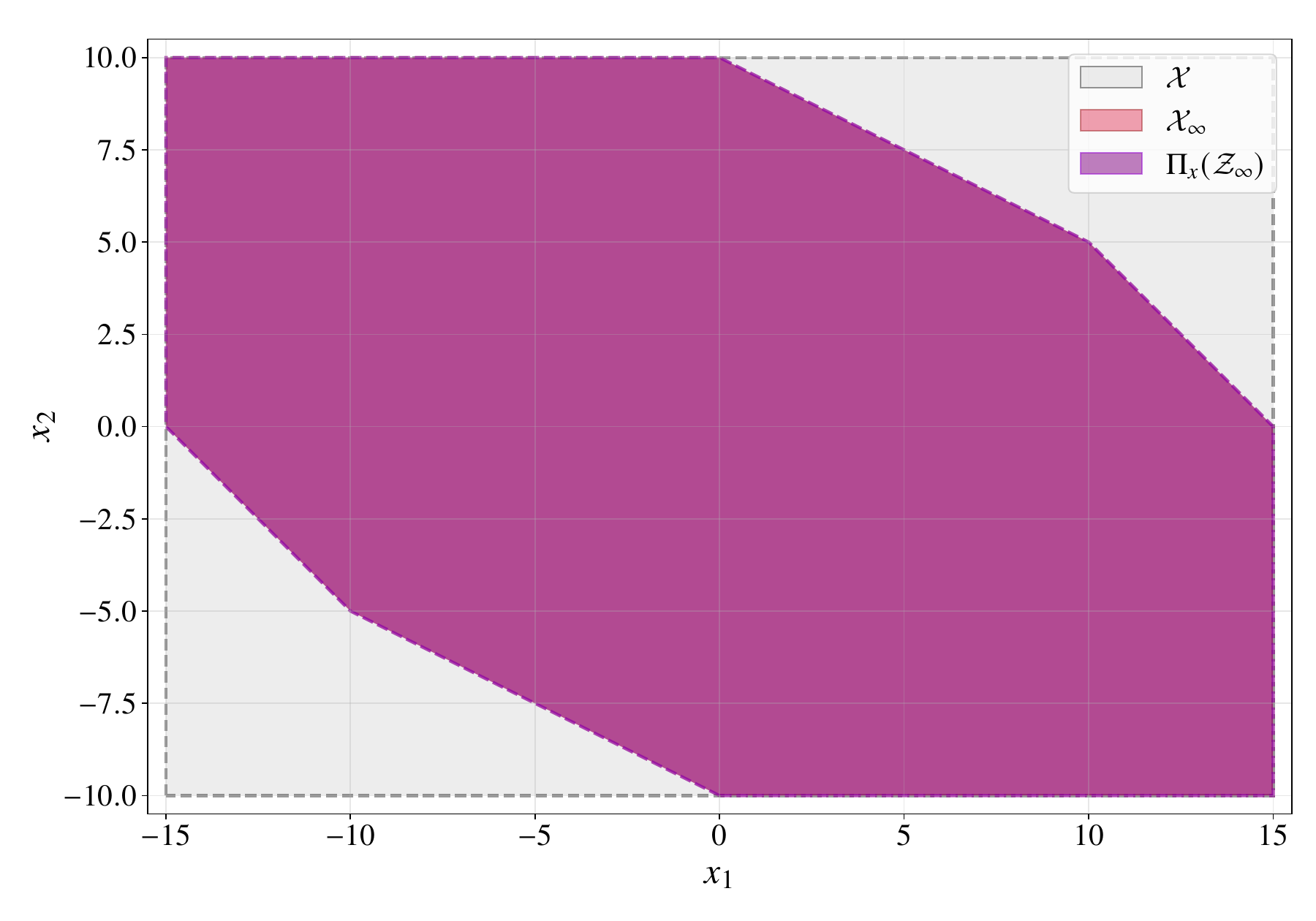}
    \caption{}\label{img:mci-z-projected}
  \end{subfigure}
  \caption{(a) Maximal state-control invariant set $\mathcal{Z}_\infty$.
  (b) Its state projection compared with the maximal control invariant set;
  the two boundaries coincide.}
\end{figure}

\begin{remark}[Controller design]\label{rem:controller-design}
The design of the controller $\pi^{\,\ell}$ is addressed in the journal extension. We only require that the controller $\pi^\ell$ produces a failing trajectory and generates enough samples so that the condition $\operatorname{rank}(Z)=p$ holds. 
\end{remark}

\section{Numerical Results}
\label{sec:results}

To validate the proposed framework, we conduct numerical experiments on the following discrete-time LTI system
\begin{equation*}
    x(k+1) = \begin{bmatrix} 1 & 1 \\ 0 & 1 \end{bmatrix}x(k) + \begin{bmatrix} 0 \\ 1\end{bmatrix} u(k), \,\, x(0) = x_0, \,\,\forall k\in\overline{\mathbb{Z}}_{+},
\end{equation*}
which can be cast in the form of \eqref{eq:lti} with $A = \begin{bmatrix} 1 & 1 \\ 0 & 1 \end{bmatrix}$, $B=\begin{bmatrix} 0 \\ 1\end{bmatrix}$, $x = [x_1,x_2]^\top$, $n_x = 2$ and $n_u = 1$, so that $p = n_x + n_u = 3$. The state and input constraints are
\begin{equation*}
    |x_1| \leq 15, \quad |x_2| \leq 10, \quad |u| \leq 5,
\end{equation*}
yielding an initial polytope $\mathcal{P}_0 \subset
\mathbb{R}^3$ with $n_\mathrm{c} = 6$ constraints. The matrices $(A,B)$ are
assumed unknown to Algorithm~\ref{alg:fail-learn} and are used only to simulate the
closed-loop system and to compute $\mathcal{Z}_\infty$ for
comparison. All experiments were conducted on a laptop with an Intel i9-11980HK CPU and an NVIDIA RTX 3080 GPU.
 
\subsection{Computing $\mathcal{Z}_\infty$ using predecessor recursion}

Using the dynamics $(A,B)$, we compute the MSCI $\mathcal{Z}_\infty$ via the recursion
\eqref{eq:mci-z-iteration}. The iteration is initialized with
$\Omega_0 = \mathcal{P}_0$. For validation, we compare the MSCI computation with the MCI computation described in~\cite{borrelli2017predictive}. 

The computed $\mathcal{Z}_\infty \subset \mathbb{R}^3$, shown in Fig.~\ref{img:mci-z}, has $14$
halfspace constraints: the initial $6$ halfspace constraints and $8$ halfspaces from the recursion \eqref{eq:mci-z-iteration}. In comparison, the MCI $\mathcal{X}_\infty$ has $8$ constraints: the initial $4$ halfspace state constraints and $4$ halfspaces from the recursion. We validate that $\Pi_x(\mathcal{Z}_\infty) = \mathcal{X}_\infty$ with a Hausdorff distance of zero as is shown in Fig.~\ref{img:mci-z-projected}, which confirms Lemma~\ref{lemma:z-recovery}(i). It takes $0.26$ seconds to compute the MSCI in $\mathbb{R}^3$ in three iterations, while it takes $0.19$ seconds to compute the MCI in $\mathbb{R}^2$ in two iterations.

\subsection{Learning $\mathcal{Z}_\infty$ from failures}

\begin{figure}[t]
  \begin{subfigure}{0.49\linewidth}
  \centering
  \includegraphics[trim=3.5cm 0.8cm 3.5cm 0.2cm, clip, width=\textwidth]{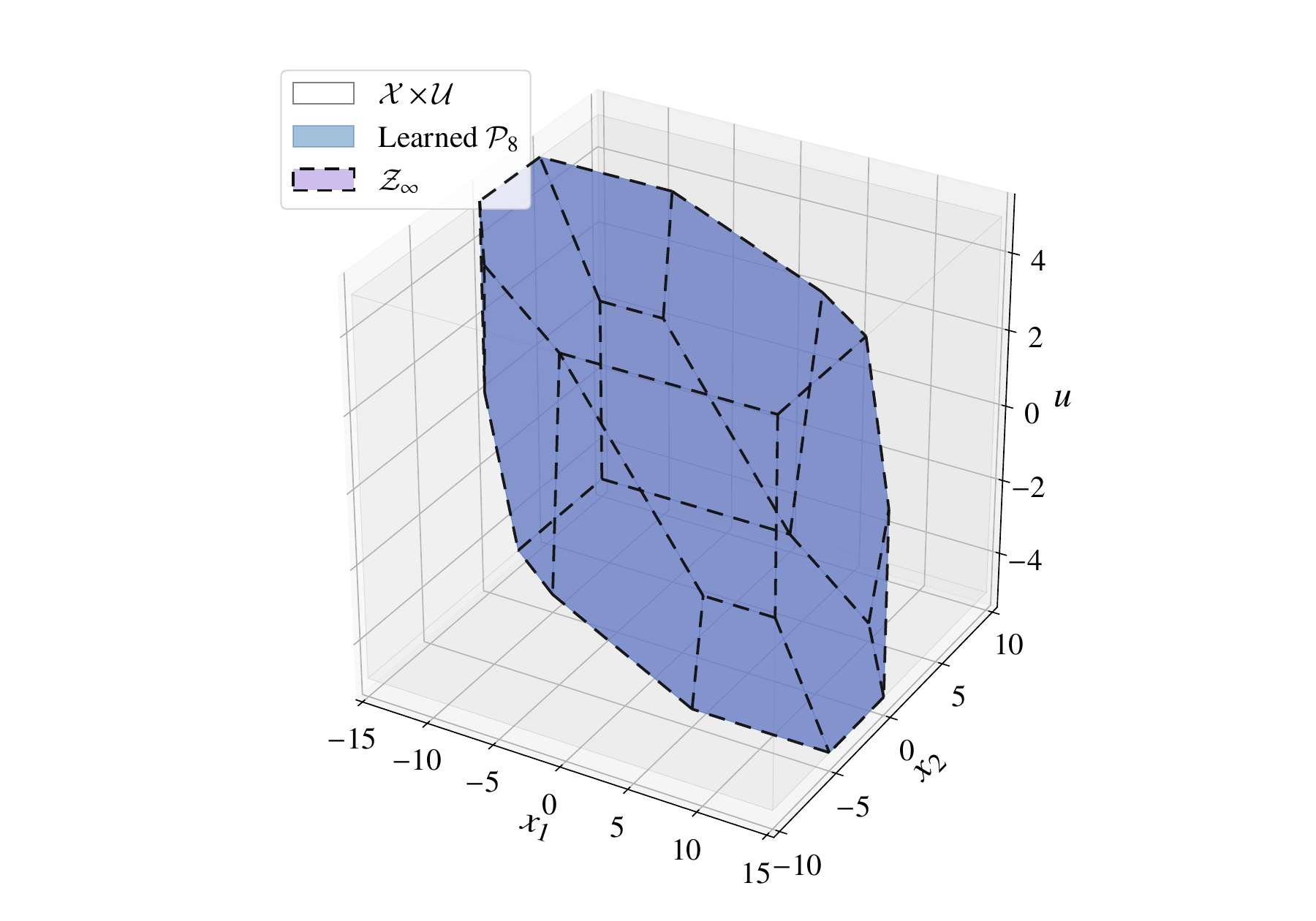}
\end{subfigure}
  \hfill
  \begin{subfigure}{0.49\linewidth}
  \centering
  \includegraphics[width=\textwidth]{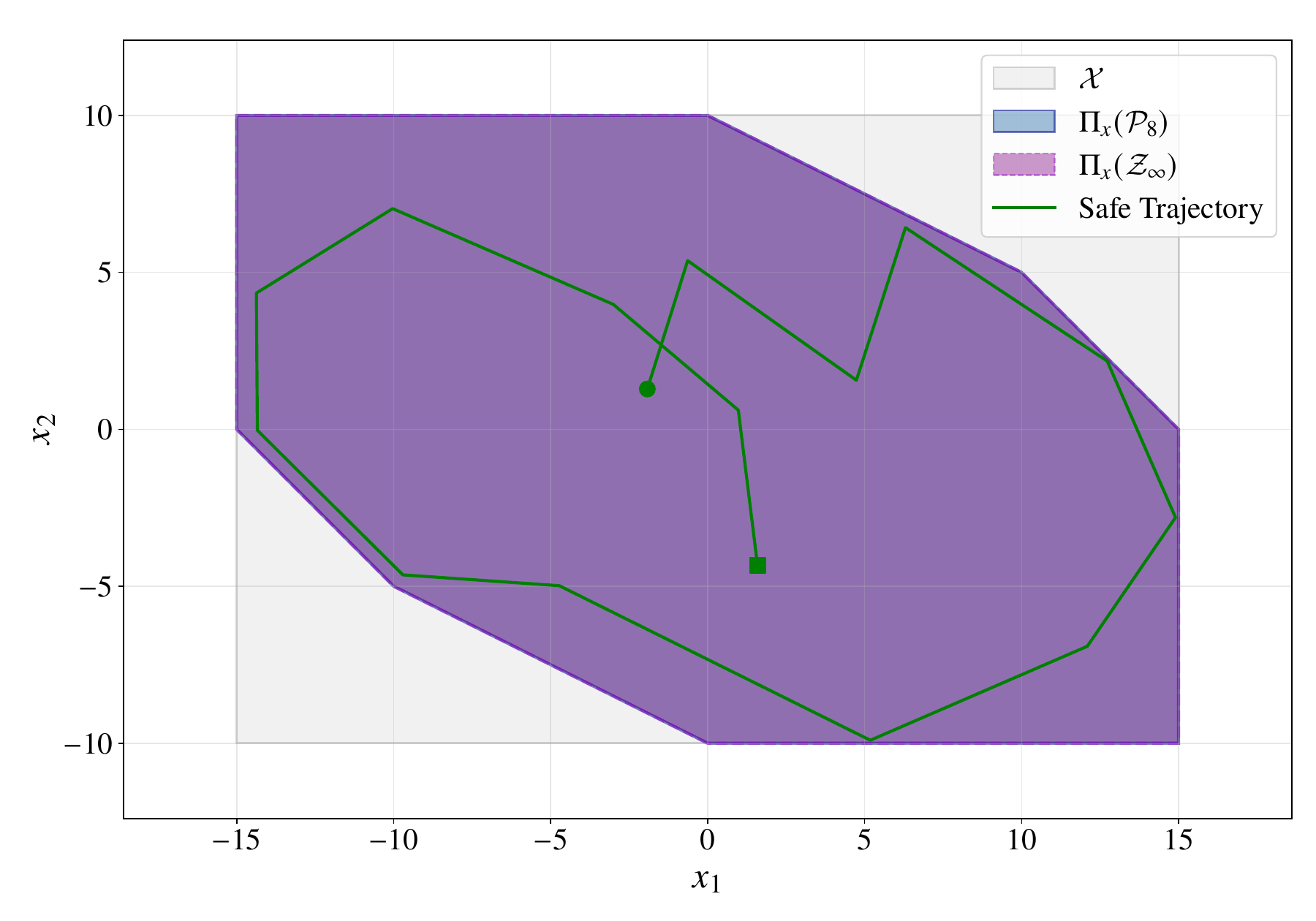}
\end{subfigure}
  \caption{The final learned polytope $\mathcal{P}_8$ (left) and its state-projection $\Pi_x(\mathcal{P}_8)$ (right). Note that $\mathcal{P}_8$ coincides with $\mathcal{Z}_\infty$ while $\Pi_x(\mathcal{P}_8)$ coincides with $\mathcal{X}_\infty$.}
  \label{fig:fail-comparison}
  % \vspace{-1.6em}
\end{figure}

We now validate Algorithm~\ref{alg:fail-learn} with $(A,B)$ unknown.
The algorithm is initialized with $\mathcal{P}_0$ and uses a sequence of failing trajectories generated by different controllers.

\textbf{Control Policies.}
Two types of controllers are used to generate failures.
First, two open-loop constant-input controllers, $u(k) \equiv u_{\min} = -5$ and $u(k) \equiv u_{\max} = 5$, are applied with $x_0 = \begin{bmatrix} 0 & 0 \end{bmatrix}^\top$ for $T = 15$ steps each. Next, trajectories are generated by a random admissible controller for $T = 15$ steps each. At each iteration $\ell \in \mathbb{Z}^+$, the initial state $x_0$ is sampled from the current projected polytope $\Pi_x(\mathcal{P}_{\ell-1})$, then, at each time step $k \in \overline{\mathbb{Z}}_{+}$, the control input $u(k)$ is sampled uniformly from $\mathcal{U}_{\ell-1}(x(k))$. In both cases, the trajectories are terminated when the state first exits $\Pi_x(\mathcal{P}_{\ell-1})$.

\textbf{Convergence of $\mathcal{P}_L$ to $\mathcal{Z}_\infty$.}
From $6$ failing trajectories, FAIL learns the $8$ predecessor halfspaces added to the initial $6$ constraints defined by $\mathcal{P}_0$, recovering all $14$ halfspaces of $\mathcal{Z}_\infty$ in $8$ iterations. It takes $0.0353$ seconds to learn the $8$ halfspaces. For every iteration $\ell$, we validate that $\mathcal{Z}_\infty \subseteq \mathcal{P}_\ell$ and that each learned halfspace is a halfspace constraint of $\mathcal{Z}_\infty$ to within numerical tolerance ($<10^{-3}$), as guaranteed by Lemma~\ref{lemma:exact-recovery}. Therefore, we have that $\mathcal{P}_L = \mathcal{Z}_\infty$ at termination. To terminate Algorithm~\ref{alg:fail-learn}, we check that no trajectory under the random admissible control policy exits $\Pi_x(\mathcal{P}_\ell)$, testing up to $1200$ trajectories at each iteration $\ell$. Fig.~\ref{fig:fail-random-no-exit} shows the $1200$ trajectories at the termination iteration $L = 8$, none of which exit $\Pi_x(\mathcal{P}_L)$. Finally, we validate that $\mathcal{P}_L$ converges to $\mathcal{Z}_\infty$ with a Hausdorff distance of 0, consistent with Theorem~\ref{theorem:monotone-convergence}. Note that this also shows that $\mathcal{U}_L(x) = \mathcal{U}^f_\infty(x) \cap \mathcal{U}(x)$, and that $\mathcal{U}^f_\infty(x) \cap \mathcal{U}(x)$ is the set of invariance-preserving inputs, consistent with Theorem~\ref{theorem:monotone-convergence} and Lemma~\ref{lemma:z-recovery} (ii).

\textbf{Extracting $\mathcal{X}_\infty$.}
Fig.~\ref{fig:fail-comparison} compares the state projection of the learned polytope $\Pi_x(\mathcal{P}_L)$ with $\mathcal{X}_\infty$. The state projection $\Pi_x(\mathcal{P}_L)$ coincides with $\mathcal{X}_\infty$, with a Hausdorff distance of zero. This demonstrates that FAIL also learns the MCI. 

\textbf{Scalability to higher-dimensional systems.} Scaling FAIL to higher-dimensional systems is limited by the computational cost of the projection $\Pi_x(\mathcal{P}_\ell)$ and is addressed in the journal extension of this paper.

\section{Conclusion}
\label{sec:conclusion}
In this paper, we introduced the concept of state-control invariance and developed a failure-aware iterative learning algorithm for computing the MSCI from observed failures. We proved that the MSCI $\mathcal{Z}_\infty$ simultaneously yields the MCI set through its state projection and the admissible invariance-preserving inputs through its $x$-sections. The FAIL algorithm learns the halfspaces defining $\mathcal{Z}_\infty$ from one-step failing state-input pairs, converging monotonically without requiring system identification. Future research will focus on extending the framework to LTI systems with Gaussian noise.

\begin{figure}[!tb]
  \centering
  \includegraphics[width=0.95\linewidth]{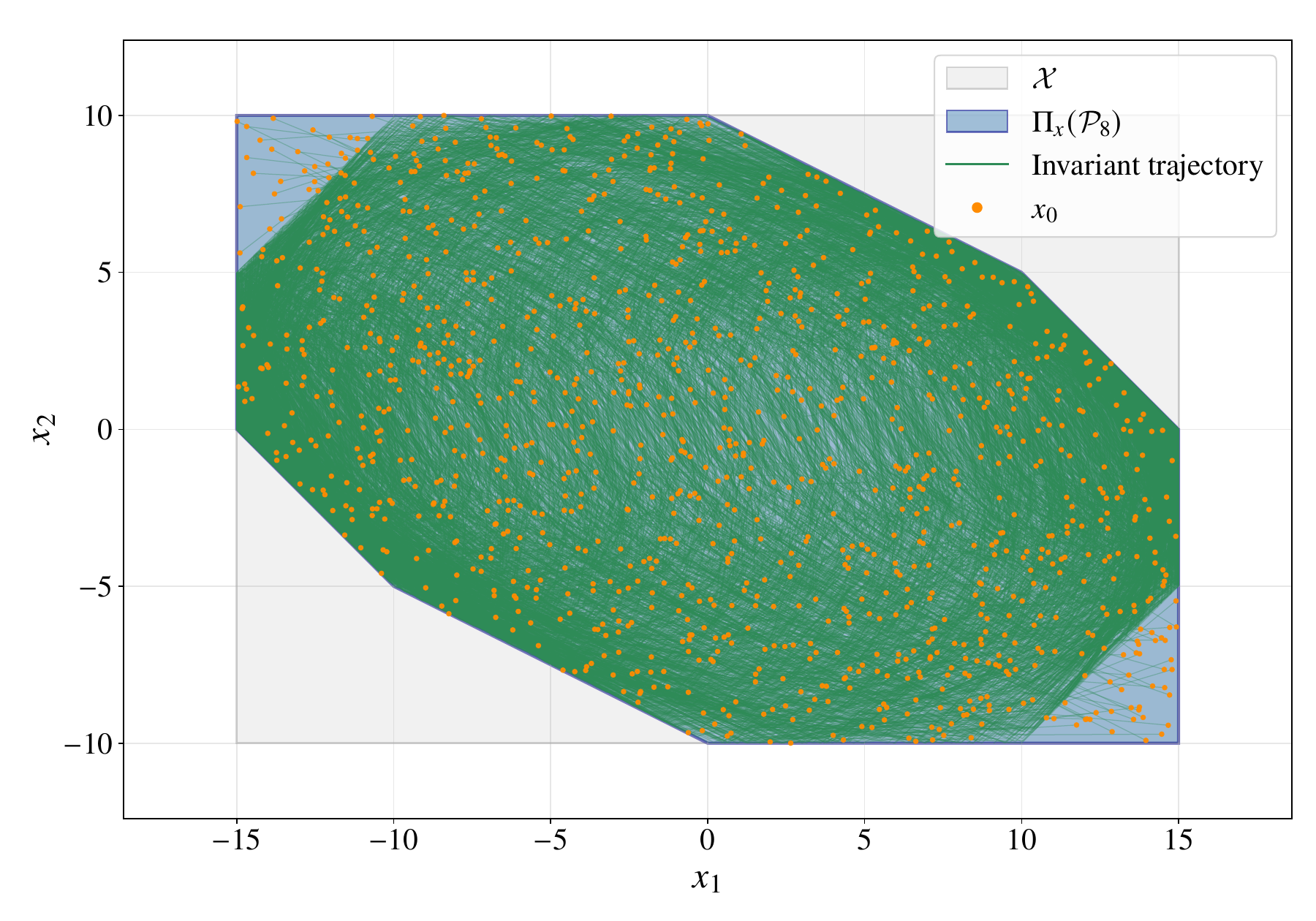}
  \caption{Non-failing trajectories generated by the random admissible controller after the final iteration $L = 8$. Each trajectory remains within $\Pi_x(\mathcal{P}_8)$.}
  \label{fig:fail-random-no-exit}
  \vspace{-1.6em}
\end{figure}

\section*{GenAI Disclosure}
The Claude chatbot by Anthropic~\cite{claude2025} was used to improve the syntax and grammar of the manuscript.
  % \vspace{-0.6em}

\bibliographystyle{IEEEtran}
\balance
\bibliography{references}             % bib file to produce the bibliography

\end{document}